\documentclass[12pt]{interact}
\usepackage{amsmath,latexsym,amsthm}
\usepackage{graphics,graphicx,subfigure,amssymb,adjustbox,a4wide}

\def\xsum{\mathop{\sum\nolimits'}}
\usepackage{color}

\newtheorem{thm}{Theorem}[section]
\newtheorem{defi}[thm]{Definition}

\newtheorem{prop}[thm]{Proposition}

\newtheorem{lemma}[thm]{Lemma}

\theoremstyle{definition}
\newtheorem{remark}[thm]{Remark}
\newtheorem{remarks}[thm]{Remarks}

\newtheorem{examples}[thm]{Examples}

\newcommand{\R}{\mathbb R}

\newcommand{\Z}{\mathbb Z}
\newcommand{\N}{\mathbb N}

\begin{document}

\title{Structural transitions in interacting lattice systems}

\author{
\name{Laurent B\'etermin}
\affil{Institute Camille Jordan, Universit\'e Claude Bernard Lyon 1,
Villeurbanne, France} \vskip0.3truecm
\name{Ladislav \v{S}amaj \and Igor Trav\v{e}nec} 
\affil{Institute of Physics, Slovak Academy of Sciences, 
D\'ubravsk\'a cesta 9, 84511 Bratislava, Slovakia}}

\date{Received:   / Accepted: }

\maketitle

\begin{abstract}
We consider two-dimensional systems of point particles located on rectangular
lattices and interacting via pairwise potentials.
The goal of this paper is to investigate the phase transitions (and their
nature) at fixed density for the minimal energy of such systems.
The 2D rectangle lattices we consider have an elementary cell of sides $a$
and $b$, the aspect ratio is defined as $\Delta=b/a$ and
the inverse particle density $A = a b$; therefore, the ``symmetric'' state
with $\Delta=1$ corresponds to the square lattice and the ``non-symmetric''
state to the rectangular lattice with $\Delta\ne 1$.
For certain types of the interaction potential, by changing continuously
the particle density, such lattice systems undertake at a specific
value of the (inverse) particle density $A^*$ a structural transition from
the symmetric to the non-symmetric state.
The structural transition can be either of first order ($\Delta$
unstick from its symmetric value $\Delta=1$ discontinuously) or of
second order ($\Delta$ unstick from $\Delta=1$ continuously); the
first and second-order phase transitions are separated by the so-called
tricritical point.
We develop a general theory on how to determine the exact values of
the transition densities and the location of the tricritical point. 
The general theory is applied to the double Yukawa and
Yukawa-Coulomb potentials.
\end{abstract}
  

\begin{amscode}
74G65, 74N05, 82B26
\end{amscode}

\renewcommand{\theequation}{1.\arabic{equation}}
\setcounter{equation}{0}

\section{Introduction} \label{Sec1}

\renewcommand{\theequation}{2.\arabic{equation}}
\setcounter{equation}{0}

Lattice systems of interacting particles are known as good models
to understand physical
\cite{Bet16,BetFriedStefEAM21,LuoRenWei20,LuoWei22,SandSerf12} or biological
\cite{Betgrid21,Mogilner03} phenomena in the simplest periodic setting
(see also \cite{BlancLewin15} and references therein).
In particular, the ground state approach -- where a lattice energy has
to be minimized in order to find the most stable state of the system --
is of high interest in applied mathematics.

\medskip

Considering a radially symmetric interaction potential $f$ with parameters, a set of lattices $\mathcal{L}\subset \R^d$, and defining the energy of this lattice system as
$$
\forall L\in \mathcal{L},\quad  E[L]=\frac{1}{2}\sum_{p\in L\backslash\{0\}} f(|p|),
$$
where $|\cdot|$ is the Euclidean norm on $\R^2$, one can ask the following important questions:
\begin{itemize}
\item[•] what is the minimizer of $E$ in the class $\mathcal{L}$?
\item[•] does this minimizer changes when parameters of $f$ vary?
\item[•] does the minimizer changes when the lattice's density varies? 
\end{itemize}

These questions have recently received a certain interest in the mathematical
community \cite{Bet16,BetPetrache19,BetFriedStefEAM21,LuoWeiLJ22,LuoWei23}
as well as their more general associated problem known as crystallization
(see \cite{BetLucPet21} and references therein), best packings
\cite{Viazetal17,Viaz17} and universal optimality \cite{CohnKumar07,Viazetal22}.
The questions of phase transitions, i.e. studying the energy minimizers
as density or parameters vary, have been shown as both interesting
\cite{Bet18,SamajTravenec19} and difficult
\cite{BetSamTrav23,LuoWeiLJ22} to solve.

\medskip

Very few results exist in that direction.
The Lennard-Jones system -- i.e. for a difference of inverse power laws --
is now entirely understood in dimension 2 \cite{Bet23,LuoWeiLJ22} in terms
of minimizers at fixed density whereas many problems in the field are still
open, with a lack of general theory.
Furthermore, only few -- but very important -- results are known
in higher dimensions, especially in dimensions 8 and 24 where best packing and
universal optimality are proven \cite{BetDefects21,Viazetal22}.
The rest of our knowledge is restricted to numerical investigations
(see e.g. \cite{BetSamTrav23}).

\medskip

The goal of this paper is to derive technical, theoretical and numerical
results for phase transition problems in lattice systems. More precisely, we
are tackling here both first and second order phase transitions for smooth
potentials $f$ among two-dimensional rectangular lattices, splitting ground
states into square and non-square structures. Following Landau's free energy
approach in statistical physics \cite{Bausch72,Landau37,Landau13,Toledano87},
we expand our energy $E$ in terms of a (small) lattice parameter and
we study the corresponding transition points (second-order phase transition),
i.e. the value of the inverse density where transition occurs.
Furthermore, we focus on important transition points called tricritical points
where (continuous) second-order transition becomes of (discontinuous)
first-order.

\medskip

Our findings are both theoretical and numerical. After defining the classes of interaction potentials and lattices we are studying in this work, we obtain the following results:
\begin{enumerate}
\item \textbf{Theoretically}, we show how a Taylor expansion of the energy $E$ can be performed in terms of the lattice parameter $\varepsilon$ ($\varepsilon=0$ corresponding to the square lattice). Results are derived ensuring the existence of transition/tricritical points and showing the universal behavior of the energy minimizer in the neighborhood of these points.\\
\item \textbf{Numerically}, we are investigating both double Yukawa and Yukawa-Coulomb potentials. Since they are highly inhomogeneous, a numerical study is performed in order to plot phase diagrams, second-order phase transition curves and tricritical points. 
\end{enumerate}
In both cases, asymptotics results in the neighborhood of
transition/tricritical points as well as singular values of
potential's parameters are derived.

\medskip

The method can be generalized easily to other types of lattices,
like the 2D (equilateral) triangle lattice which is the centered rectangle
lattice with the aspect ratio $\sqrt{3}$ against the general centered
rectangle lattice.
The generalization of the method to 3D lattice structure is also
straightforward.
The application of the theory to other types of interaction potentials,
especially to Lennard-Jones potential, is of our future interest.

\medskip

\textbf{Plan of the paper.}
In Section \ref{Sec2}, we give the precise definitions of potentials,
lattices and energies we are considering.
General results on phase transitions are proved in Section \ref{Sec3}
whereas applications to double Yukawa and Yukawa-Coulomb potentials are
presented in Sections \ref{Sec4} and \ref{Sec5}, respectively,
with both numerical and theoretical aspects.

\section{Preliminary formalism} \label{Sec2}
In this section, we briefly present the type of potential,
lattices and energy we are considering in this paper.

\medskip

Let us start with potentials. Our goal is to cover the main interaction potentials presented for instance in \cite{Kaplan} (see also \cite{Bet16}).
\begin{defi}[\textbf{Admissible potential}]
We say that $f\in \mathcal{F}$ if $f:(0,\infty)\to \R$, $|f(r^2)|=O(r^{-2-\varepsilon})$ for some $\varepsilon>0$ as $r\to \infty$, $f\in C^{\infty}(\R_+^*)$ and there exists a Radon measure $\mu_f$ on $(0,+\infty)$ such that
\begin{equation}\label{fgen}
f(r) = \int_0^\infty e^{-r^2 t}\ {\rm d}\mu_f(t).
\end{equation}
Furthermore, we say that $f\in \mathcal{F}_+$ if $\mu_f$ is non-negative.
\end{defi}
\begin{remarks}[\textbf{Completely monotone potentials}]
Let us write $f(r)=F(r^2)$ where $F:(0,+\infty)\to \R$, then:
\begin{enumerate}
\item the measure $\mu_f$ is actually the inverse Laplace transform of $F$;
\item By Bernstein-Hausdorff-Widder theorem \cite{Bernstein29}, we know that $f\in \mathcal{F}_+$ if and only if $F$ is completely monotone, i.e. the derivatives of $F$ alternate their sign: $\forall k\in \N$, $\forall r>0$, $(-1)^{k}F^{(k)}(r)\geq 0$.
\end{enumerate}
\end{remarks}

\begin{examples}[\textbf{Riesz and Yukawa potentials}]
Let us mention two important interaction potential we will consider
in this work:
\begin{enumerate}
\item the Riesz potential with parameter $s>0$ is given by
\begin{equation}\label{riesz}
f(r) = \frac{1}{r^s} , \qquad
d\mu_f(t) =  \frac{1}{\Gamma(s/2)} t^{s/2-1}dt
\end{equation}
with $\Gamma$ being the Euler Gamma function.
We notice that $f\in \mathcal{F}_+$ if and only if $s>2$.
Nevertheless, we will also consider lower exponents by renormalizing
our lattice energy (see Remark \ref{rmk-renorm}) even though
$f\not\in \mathcal{F}$ in that case, because of its non-integrability
at infinity.
\item the Yukawa potential, belonging to $\mathcal{F}_+$ reads,
for $\kappa>0$, as
\begin{equation}\label{yukawa}
f(r) = \frac{e^{-\kappa r}}{r} , \qquad
d\mu_f(t) =  \frac{1}{\sqrt{\pi t}}
\exp\left(-\frac{\kappa^2}{4t} \right)dt .
\end{equation}
\end{enumerate}
\end{examples}

The set of lattice structures we are considering is defined as follows.

\begin{defi}[\textbf{Family of rectangular lattices}]
Let $\Delta \in (0,1]$ and $A>0$. The rectangular lattice of area $A$ and side-lengths $\sqrt{A\Delta}$ and $\sqrt{\frac{A}{\Delta}}$ is defined as
$$
L_{\Delta, A}:=\sqrt{A}\left[\Z\left(\frac{1}{\sqrt{\Delta}},0 \right)\oplus \Z \left(0,\sqrt{ \Delta} \right)\right],
$$
and its associated quadratic form is defined by
$$
\forall (j,k)\in \Z^2,\quad Q_{\Delta, A}(j,k):=\left|j\sqrt{A}\left(\frac{1}{\sqrt{\Delta}},0 \right)+k\sqrt{A}\left(0,\sqrt{ \Delta} \right) \right|^2= A \left( \frac{j^2}{\Delta} + k^2\Delta \right),
$$
where $|\cdot|$ is the euclidean norm on $\R^2$.
\end{defi}
\begin{remark}
Notice that this family of lattices is exactly, up to isometry, the set of all rectangular lattices with fixed density. Furthermore, the particle system is invariant under rotation by the right angle $\pi/2$
which is equivalent to the interchange
$\Delta$ and $1/\Delta$.
The special case $\Delta=1$ corresponds to the ``symmetric'' state of
the square lattice, the case $\Delta\ne 1$ corresponds to
the ``non-symmetric'' state of the rectangle lattice.
\end{remark}

The lattice energy we are studying in this paper is the energy per point of $L_{A,\Delta}$ interacting via potential $f$.

\begin{defi}[\textbf{Interaction energy}]
Let $f\in \mathcal{F}$, $A>0$, $\Delta\in (0,1]$, then the $f$-energy of $L_{A,\Delta}$ is defined by
\begin{equation} \label{ereca}
E(A,\Delta):=\frac{1}{2} \xsum_{j,k=-\infty}^{\infty}
f\left(\sqrt{Q_{A,\Delta}(j,k)} \right)=\frac{1}{2} \xsum_{j,k=-\infty}^{\infty}
f \left( \sqrt{A \left( \frac{j^2}{\Delta} + k^2\Delta \right)} \right).
\end{equation}
\end{defi}
\begin{remark}
The prefactor $1/2$ appears because each energy term is shared by two particles and the prime at the sum means that the self-energy term $j=k=0$ is omitted. Note that the dependence of the energy on the parameters of the potential $f$ is not written explicitly, for simplicity reasons.
\end{remark}

\begin{remark}\label{regularity}
Since $f$ is admissible, it follows immediatly that, by composition and absolute summability of the energy \eqref{ereca}, $E\in C^{\infty}(\R_+^*\times \R_+^*)$.
\end{remark}

As already done in other papers \cite{Bet16,Faulhuber17} for rectangular or general lattices, this energy can be written in terms of Jacobi theta functions thanks to the Laplace transform expression of $f$.

\begin{lemma}[\textbf{Integral representation of the energy}, see e.g.
\cite{Bet16,Faulhuber17}]
Let $f\in \mathcal{F}$, $A>0$, $\Delta\in (0,1]$, then we have
\begin{equation} \label{final}
E(A,\Delta) = \frac{1}{2} \int_0^\infty
\left[ \theta_3\left({\rm e}^{-t/\Delta}\right)
\theta_3\left({\rm e}^{-t\Delta}\right) - 1 \right] {\rm d}\mu_f\left(\frac{t}{A} \right) , 
\end{equation}
where
\begin{equation} \label{theta}
\theta_3(q)=\sum_{j=-\infty}^{\infty} q^{j^2}
\end{equation}
denotes Jacobi elliptic function with zero argument.
\end{lemma}
\begin{remark}
We use the Gradshteyn-Ryzhik \cite{Gradshteyn} notation for the third Jacobi theta function. Furthermore the subtraction of $-1$ in the square bracket is due to the absence of
the self-energy term $(j,k)=(0,0)$ in the sum (\ref{ereca}). If $\mu_f$ is absolutely continuous with respect to the Lebesgue measure, i.e. $d\mu_f(t)=\rho_f(t)dt$, then we have
\begin{equation}
{\rm d}\mu_f\left(\frac{t}{A} \right) = \rho_f\left(\frac{t}{A} \right) \frac{{\rm d}t}{A} . 
\end{equation}  
Moreover, note that the invariance of $E(A,\Delta)$ with respect to the transform
$\Delta\to 1/\Delta$ is obviously ensured by the formula (\ref{final}).
\end{remark}

\begin{remark}\label{rmk-renorm}
If the potential $f\not\in \mathcal{F}$ decays to zero at large distances too slowly -- like for instance for the Riesz potential with $s<2$ --
the integral in (\ref{final}) may diverge which requires a regularization. Typical examples are available in the literature, in particular in \cite[Section IV.2]{LewinJelliumReview}, where the renormalization for Riesz potential and logarithmic energy are presented and where it is showed that adding a uniform neutralizing background to the system yields to the analytic continuation of the Epstein zeta function. In the Coulomb case with $s=1$, it corresponds to the subtraction of the singular
term $-\pi/t$ in the square bracket in the integral (\ref{final}),
as it is also explained in \cite{Travenec22}. Therefore, the energy for 
$$
f(r)=\frac{1}{r}=\int_0^\infty e^{-r^2 t} \frac{dt}{\sqrt{\pi}\sqrt{t}},
$$ 
has to be redefined, using a straightforward change of variable, as
$$
E(A,\Delta)=\frac{1}{2\sqrt{A\pi}}\int_0^\infty\left[\theta_3(e^{-t\Delta})\theta_3(e^{-t/\Delta})-1-\frac{\pi}{t} \right]\frac{dt}{\sqrt{t}}.
$$
\end{remark}
\renewcommand{\theequation}{3.\arabic{equation}}
\setcounter{equation}{0}

\section{General theory of structural transitions} \label{Sec3}
To account better for the symmetry $\Delta\to 1/\Delta$ of the energy 
(\ref{final}), we introduce the parameter $\varepsilon$ via
\begin{equation} \label{Delta}
\Delta = e^\varepsilon .
\end{equation}
Thus the $\Delta\to 1/\Delta$ symmetry of the energy (\ref{final}) is
converted to the $\varepsilon\to -\varepsilon$ one of the energy
\begin{equation} \label{finall}
E(A,e^\varepsilon) = \frac{1}{2} \int_0^\infty 
\left[ \theta_3\left({\rm e}^{-t\exp(-\varepsilon)}\right)
\theta_3\left({\rm e}^{-t\exp(\varepsilon)}\right) - 1 \right]{\rm d}\mu_f\left(\frac{t}{A} \right) . 
\end{equation}
At the same time, the symmetric (square lattice) value of the parameter
$\Delta=1$ is consistent with $\varepsilon=0$. 

\medskip

Let us recall the universal optimality result (among lattices) concerning the square lattice, due to Montgomery \cite{Montgomery88}.

\begin{prop}[\textbf{Universal optimality among rectangular lattices}]
If $f\in \mathcal{F}_+$, then $\varepsilon=0$ is the unique minimizer of $\varepsilon\mapsto E(A,e^\varepsilon)$ for all $A>0$.
\end{prop}
\begin{proof}
Since $\mu_f$ is nonnegative and $\Delta=1$ is the unique minimizer of the positive function $\Delta\mapsto \theta_3(e^{-t\Delta})\theta_3(e^{-t\Delta^{-1}})-1$ as shown in \cite{Montgomery88}, it follows that for all $\Delta>0$, $E(A,\Delta)\geq E(A,1)$ with equality if and only if $\Delta=1$. Applying the change of variable $\Delta=\exp(\varepsilon)$ completes the proof.
\end{proof}

The following result gives the expansion of $E(A,e^\varepsilon)$, for fixed $A$ and as $\varepsilon\to 0$.

\begin{thm}[\textbf{Taylor expansion of the energy}]
Let $f\in \mathcal{F}$ and $A>0$, then, as $\varepsilon\to 0$,
\begin{equation} \label{ecreps}
E(A,e^\varepsilon) = E_0(A) + E_2(A) \varepsilon^2 + E_4(A) \varepsilon^4
+ O(\varepsilon^6) ,
\end{equation}
where 
\begin{equation} 
E_0(A) = E(A,1)= \frac{1}{2} \int_0^\infty 
\left[ \left( \theta_3 \right)^2 - 1 \right]{\rm d}\mu_f\left(\frac{t}{A} \right)
\end{equation}
is the energy of the square lattice,
\begin{equation} \label{E2A} 
E_2(A) =\frac{d^2}{d\varepsilon^2}\left[ E(A,e^\varepsilon)\right]_{|\varepsilon=0}
= \frac{1}{2} \int_0^\infty \left[ t \theta_3 \theta_3^{(1)}
- t^2 \left( \theta_3^{(1)} \right)^2 + t^2 \theta_3 \theta_3^{(2)} \right]
{\rm d}\mu_f\left(\frac{t}{A} \right) ,
\end{equation}
and 
\begin{eqnarray} 
E_4(A)  = \frac{d^4}{d\varepsilon^4}\left[ E(A,e^\varepsilon)\right]_{|\varepsilon=0}
& = & \frac{1}{24} \int_0^\infty 
\left[ t \theta_3 \theta_3^{(1)} - t^2 \left( \theta_3^{(1)} \right)^2
+ 7 t^2 \theta_3 \theta_3^{(2)} + 6 t^3 \theta_3 \theta_3^{(3)}
-6 t^3 \theta_3^{(1)} \theta_3^{(2)}\right. \nonumber \\ & & \left.
 + t^4 \theta_3 \theta_3^{(4)}
- 4 t^4 \theta_3^{(1)} \theta_3^{(3)} + 3 t^4 \left( \theta_3^{(2)} \right)^2
\right]{\rm d}\mu_f\left(\frac{t}{A} \right),
\end{eqnarray}
and where the theta function and its derivatives are written for simplicity as
\begin{equation}
\theta_3 := \theta_3\left({\rm e}^{-t}\right) , \qquad
\theta_3^{(n)} := \frac{{\rm d}^n}{{\rm d}t^n}
\theta_3\left({\rm e}^{-t}\right), n\in \N.
\end{equation} 
\end{thm}
\begin{remark}
Note that because of the $\varepsilon\to -\varepsilon$ symmetry of the energy,
the expansion contains only {\em even} powers of $\varepsilon$.
\end{remark}
\begin{proof}
We use the fact  that $f\in C^{\infty}(0,\infty)$ and $f$ is absolutely summable
at infinity which implies that $E(A,\cdot)\in C^\infty(0,\infty)$ (see Remark \ref{regularity}).
Therefore, the wished expansion easily yields from the one of
the theta functions product, i.e. for all $t>0$ and as $\varepsilon\to 0$,
the quantity $\theta_3\left({\rm e}^{-t\exp(-\varepsilon)}\right)
\theta_3\left({\rm e}^{-t\exp(\varepsilon)}\right)$ is given by
\begin{eqnarray}
\sum_{j,k} {\rm e}^{-j^2 t\exp(-\varepsilon)} {\rm e}^{-k^2 t\exp(\varepsilon)} & = &
\sum_{j,k} {\rm e}^{-(j^2+k^2)t} \Big[ 1 + \left( j^2-k^2 \right) t \varepsilon
\nonumber \\ & &
+ \frac{1}{2} \left( -j^2-k^2-2j^2 k^2 t + j^4t + k^4 t \right) t \varepsilon^2
\nonumber \\ & &
+\frac{1}{6} \left( j^2-k^2 \right) \left( 1 - 3j^2 t - 3k^2 t
-2 j^2 k^2 t^2 +j^4 t^2 + k^2 t^2 \right) t \varepsilon^3
\nonumber \\ & &
+\frac{1}{24} \left( - j^2 - k^2 -2 j^2 k^2 t + 7 j^4 t + 7 k^4 t
- 6 j^6 t^2 - 6 k^6 t^2 \right. \nonumber \\ & & \left. 
+ 6 j^4 k^2 t^2 + 6 j^2 k^4 t^2 + j^8 t^3 + k^8 t^3 - 4 j^6 k^2 t^3
\right. \nonumber \\ & & \left. - 4 j^2 k^6 t^3 + 6 j^4 k^4 t^3 \right)
t \varepsilon^4 + O\left( \varepsilon^6\right) \Big] . 
\end{eqnarray}
The terms of odd orders $\varepsilon$ and $\varepsilon^3$ vanish
because of the antisymmetry of summands with respect to the interchange
of indices $j$ and $k$.
The terms of even orders $\varepsilon^2$ and $\varepsilon^4$ can be further
simplified by using the interchange of indices $j$ and $k$ and the sums
of type $\sum_j j^{2n} \exp(- j^2 t)$ with $n\in \N^*$ are equal to
$(-1)^n \theta_3^{(n)}$.
\end{proof}

\begin{remark}[\textbf{Connection with Statistical Physics}]
The present exact expansion (\ref{ecreps}) has its counterpart
in statistical physics where it represents a mean-field approximation for
a complicated statistical model in thermal equilibrium at some temperature,
known as the Landau free energy \cite{Bausch72,Landau37}.
The parameter $\varepsilon$ plays there the role of the order parameter
which vanishes in the disordered phase and takes non-zero values in
the ordered phase.
The general analysis of the expansion (\ref{ecreps}) in the context of
critical phenomena in statistical models can be found in many textbooks,
see e.g. \cite{Landau13,Toledano87}.
In what follows, we shall adopt this general analysis to our ground-state
problem.
\end{remark}

By (\ref{ecreps}), we have that, for all $A>0$, as $\varepsilon\to 0$,
$$
E(A,e^\varepsilon)-E(A,1) = E_2(A)\varepsilon^2 + E_4(A)\varepsilon^4
+O(\varepsilon^6) = \varepsilon^2\left[ E_2(A)+ E_4(A)\varepsilon^2
+ O(\varepsilon^4)\right].
$$
Therefore, it appears that the sign of $E_2(A)$ determines the one of
$E(A,e^\varepsilon)-E(A,1)$ for sufficiently small values of $\varepsilon>0$,
giving the optimality of the square lattice when $E(A,e^\varepsilon)-E(A,1)>0$
and the optimality of a non-square one when $E(A,e^\varepsilon)-E(A,1)<0$.

\begin{defi}
Let $f\in \mathcal{F}\backslash \mathcal{F}_+$, then any $A^*$ such that 
\begin{equation}\label{transpointenlarged}
E_2(A^*)=0
\end{equation}
with a change of sign for $E_2$ at $A=A^*$ is called a transition point.
\end{defi}

A simple condition can be written for insuring the existence of such transition point.

\begin{prop}[\textbf{Existence of a transition point}] \label{prop:transition} Let $f\in \mathcal{F}$. Then:
\begin{enumerate}
\item if $f\in \mathcal{F}_+$, then there is no transition point;
\item if $f\in \mathcal{F}\backslash \mathcal{F}_+$ such that $\mu_f<0$ on $(0,r_0)$ or on $(r_1,+\infty)$ for some $r_0, r_1>0$, then a transition point exists.
\end{enumerate}
\end{prop}
\begin{proof} 
It has been shown by Faulhuber and Steinerberger \cite{Faulhuber17} that
$$
\forall t>0,\quad t \theta_3 \theta_3^{(1)} + t^2 \theta_3 \theta_3^{(2)}
- t^2 \left( \theta_3^{(1)} \right)^2>0.
$$
Therefore, if $f\in \mathcal{F}_+$, then $\mu_f$ is nonnegative and it follows that $E_2(A)>0$ for all $A>0$.\\
Furthermore, the second point follows from \cite[Prop. 4.4]{BetPetrache19} applied to the square lattice among rectangular lattices in dimension $d=2$.
\end{proof}
\begin{examples}
Point $(1)$ of the previous result can be illustrated by any Riesz potential, whereas point $(2)$ is covered for instance by the double-Yukawa potential we study in the next section of the paper, given for all $r>0$ by
\begin{equation*}
f(r) = v_1 \frac{\exp{(-\kappa_1 r)}}{r} - v_2 \frac{\exp{(-\kappa_2 r)}}{r}, \quad v_1 > v_2 >0,\qquad \kappa_1 > \kappa_2 > 0.
\end{equation*}

\end{examples}

We now give conditions such that a second order phase transition (see \cite{Samaj13}), i.e. the transition for the minimizer of $\varepsilon\to E(A,e^\varepsilon)$ is continuous.

\begin{prop}[\textbf{Second order phase transition}]\label{prop:2ndorder}
Let $f\in \mathcal{F}$ and $A^*$ be a transition point such that:
\begin{enumerate}
\item $E_2$ is strictly decreasing in the neighborhood of $A^*$;
\item $E_4(A^*)>0$.
\end{enumerate}
Then there exists $A_0>0$ and $A_1>0$ such that
\begin{itemize}
\item[•] if $A_0<A<A^*$, then $\varepsilon=0$ is the unique minimizer of $\varepsilon\to E(A,e^\varepsilon)$;
\item[•] if $A^*<A<A_1$, any local minimizer $\varepsilon$ of $\varepsilon\to E(A,e^\varepsilon)$ satisfies the following asymptotics, as $A\to A^*$:
$$
\varepsilon=\sqrt{\frac{-\frac{dE_2}{dA}(A^*)}{2E_4(A^*)}}\sqrt{A-A^*}+o(\sqrt{A-A^*}).
$$
\end{itemize}
\end{prop}
\begin{remark}
The same result can be written when $E_2$ is strictly increasing in the neighborhood of $A^*$, with the reverse condition $E_4(A^*)<0$, which simply exchanges the regimes of optimality.
\end{remark}
\begin{proof}
Since a transition point exists, a simple Taylor expansion of $E_2$ as $A\to A^*$ -- ensured by the fact that $E(\cdot,\Delta)\in C^\infty(\R_+^*)$ (see Remark \ref{regularity})-- reads
\begin{align*}
E_2(A)&=E_2(A^*)+\frac{d E_2}{d A}(A^*)(A-A^*)+o(A-A^*)\\
&=-\frac{d E_2}{d A}(A^*)(A^*-A)+o(A-A^*).
\end{align*}
By assumption, we know that $\frac{d E_2}{d A}(A^*)<0$ and let us write $b:=-\frac{d E_2}{d A}(A^*)>0$ for simplicity in such a way that, as $A\to A^*$,
\begin{equation}\label{E2}
E_2(A)=b(A^*-A)+o(A^*-A).
\end{equation}
The value of $\varepsilon$ which locally minimizes $E(A,\cdot)$ satisfies the following asymptotic stationarity condition we get from (\ref{ecreps}):
\begin{equation}\label{crit}
\frac{\partial E(A,e^\varepsilon)}{\partial \varepsilon}
= 2 E_2(A) \varepsilon + 4 E_4(A) \varepsilon^3 + O(\varepsilon^5) = 0 .    
\end{equation}
Using \eqref{E2}, this condition, for $A$ in the neighborhood of $A^*$, is therefore equivalent to the couple of equations
\begin{equation}
\varepsilon = 0 \qquad \textnormal{or} \qquad \varepsilon^2 =
\frac{b}{2 E_4(A^*)} \left( A - A^* \right) +o(A-A^*) .
\end{equation}  
We know that the first solution is dominant in a certain region
$A_0<A<A^*$ with $A_0\ge 0$.
Since $\varepsilon^2$ must be a real positive number and $E_4(A^*)>0$,
the second solution is
\begin{equation} \label{order}
\varepsilon= \sqrt{\frac{b}{2 E_4(A^*)}}
\left( A - A^* \right)^{1/2} +o(\left( A - A^* \right)^{1/2}) ,
\end{equation}  
dominant in the region $A^*<A<A_1$ for a certain $A_1>0$,
which completes the proof.
\end{proof}
\begin{remark}
Notice that indeed the transition at $A=A^*$ from $\varepsilon=0$
to $\varepsilon$ given by Eq. (\ref{order}) is continuous.
Consequently, the energy and its first derivative with respect to $A$ are
continuous as well.\\
Furthermore, the singular behaviour $\varepsilon\propto(A-A^*)^{\beta}$
as $A\to A^*$, $A>A^*$, of the minimizer (and actually any local one) defines the critical exponent $\beta$
which in our case acquires the mean-field value $1/2$.
This critical exponent is universal in the sense that it does not depend on
the particular form of the interaction potential $f$ as long as the (general) assumptions are satisfied.
\end{remark}

Let us now assume that our interaction potential $f=f_\alpha$ depends on
a parameter $\alpha\in \R$.
Thus, a transition point satisfying the assumptions of second-order
transition is $\alpha$-dependent and one can plot the graph of
$\alpha\mapsto A^*(\alpha)$.
We therefore get a curve of second-order transitions between square and
non-square rectangular lattices as minimizers of $E(A,\cdot)$.
In analogy with statistical mechanics, we expect only
one curve of second-order transition.
This curve exists say for $\alpha>\alpha^t$ and necessarily
terminates at some $\alpha=\alpha^t$ once $A^*(\alpha=\alpha^t)$ satisfies,
besides $E_2(A^*(\alpha^t))=0$, also the condition $E_4(A^*(\alpha^t))=0$
which changes sign at this point.

\begin{defi}
Let $f_\alpha\in \mathcal{F}$ depending on a real parameter $\alpha$. We say that $A^t>0$ and $\alpha^t$ are coordinates of
a tricritical point if it is a transition point satisfying both conditions
\begin{equation} \label{tricritical}
E_2(A^t(\alpha^t)) = 0 \qquad \textnormal{and} \qquad E_4(A^t(\alpha^t)) = 0,
\end{equation}
with a change of sign for $E_4$ at $A^t(\alpha^t)$.
\end{defi}

Beyond this point, i.e., for $\alpha < \alpha^t$,
transitions between the square and rectangle lattices
become of first order, i.e., there is a discontinuity in $\varepsilon$
from 0 to a finite value.
This means that our expansion (\ref{ecreps}) no longer applies and 
the curve of first-order transitions can be located only numerically
by plotting the energy in the whole interval $\Delta\ge 1$.
Using the same method as before, we can derive the singular behavior of
$\varepsilon$ at the tricritical value of $\alpha=\alpha^t$,
in the region $A>A^t$.

\begin{prop}[\textbf{Transition at the tricritical point}]
Let $f_\alpha\in \mathcal{F}$ depending on a real parameter $\alpha$ and $(\alpha^t,A^t)$ be the coordinates of
a tricritical point such that
\begin{enumerate}
\item $E_2$ is strictly decreasing function of $A$ in the neighborhood of $A^t$;
\item $E_6(A^t)=\frac{d^6}{d\varepsilon^6}
\left[E(A,e^\varepsilon)\right]_{|\varepsilon=0}>0$.
\end{enumerate}
Therefore, there exists $A_2>0$ such that for $A^t<A<A_2$, any local minimizer of
$\varepsilon\to E(A,e^\varepsilon)$ satisfies the following asymptotics as
$A\to A^t$:
$$
\varepsilon = \sqrt[4]{\frac{-\frac{d E_2}{dA}(A^t)}{3 E_6(A^t)}}
(A-A^t)^{1/4}+o\left((A-A^t)^{1/4} \right).
$$
\end{prop}
\begin{remark}
As Proposition \ref{prop:2ndorder}, the same result for the $A<A^t$ regime
can be written when $E_2$ is strictly increasing in the neighborhood of $A^t$
and $E_6(A^t)<0$.
\end{remark} 
\begin{proof}
This follows from \eqref{E2} and \eqref{crit} expanded to order $5$ --
coming from the Taylor expansion of the energy to order $6$, since $E\in C^\infty(\R_+^*\times \R_+^*)$, as
$E(A,e^\varepsilon)=E_0(A)+E_2(A)\varepsilon^2+E_4(A)\varepsilon^4
+E_6(A)\varepsilon^6+o(\varepsilon^8)$ -- and following
exactly the same lines as in the proof of
Proposition \ref{prop:2ndorder}.
\end{proof}
\begin{remark}
For $\alpha=\alpha^t$, we therefore get the singular
behavior $\varepsilon\propto (A-A^t)^{1/4}$ as $A\to A^t$, $A>A^t$, so that
the critical exponent at the tricritical point is $\beta^t=1/4$ which
is again universal up to the assumptions.
\end{remark}

The existence of transition and tricritical points and the second-order
transition curves depends on the form of the interaction potential $f$
(see e.g. Proposition \ref{prop:transition} for the transition point).
In the next two sections, we shall present explicit results for
the double Yukawa and Yukawa-Coulomb potentials which are combinations of two
(repulsive and attractive) terms.
Lattice systems with these potentials exhibit in certain regions
of model's parameters continuous as well as discontinuous
structural transition from the square to rectangular lattices.

\renewcommand{\theequation}{4.\arabic{equation}}
\setcounter{equation}{0}

\section{Double Yukawa potential} \label{Sec4}
\textbf{Potential, parameters constraints and energy.}
The double Yukawa interaction potential $f\in \mathcal{F}$ is defined on
$(0,+\infty)$ by
\begin{equation}\label{doubleyu}
f(r) = v_1 \frac{\exp{(-\kappa_1 r)}}{r} - v_2 \frac{\exp{(-\kappa_2 r)}}{r},
\quad v_1 > v_2 >0,\quad\kappa_1 > \kappa_2 > 0,
\end{equation}
where we choose the parameters in such a way that the first repulsive term
dominates at small distances whereas the second is attractive for large one.
Furthermore, the corresponding measure, given by the inverse Laplace
transform of $r\mapsto f(\sqrt{r})$, reads as
\begin{equation} \label{muft}
{\rm d}\mu_f(t) =
\left[ v_1\exp{\left(-\frac{\kappa_1^2}{4t}\right)}
-v_2\exp{\left(-\frac{\kappa_2^2}{4t}\right)}\right]
\frac{{\rm d}t}{\sqrt{\pi t}}.
\end{equation}

The potential (\ref{doubleyu}) is supposed to have an attractive minimum
$f(r_{\min})<0$ at some $r=r_{\min}$.
For simplicity, we set $r_{\min}=1$ and $f(1)=-1$, reducing in this way
the number of parameters four by two. We keep the independent parameters
$(v_1,\kappa_1)$ since
\begin{equation} \label{v2k2}
\kappa_2 = \frac{\kappa_1 v_1-{\rm e}^{\kappa_1}}{v_1+{\rm e}^{\kappa_1}}
\quad \textnormal{and}\quad
v_2 = \frac{{\rm e}^{\kappa_2-\kappa_1}(1+\kappa_1)v_1}{1+\kappa_2} .  
\end{equation} 
Since $\kappa_2>0$, the parameters $v_1$ and $\kappa_1$ are constrained
to the subspace
\begin{equation} \label{subspace}
v_1 > \frac{{\rm e}^{\kappa_1}}{\kappa_1} .  
\end{equation}

The energy per particle for the double Yukawa potential reads,
given \eqref{v2k2}, as
\begin{equation} \label{erec}
E(A,\Delta;v_1,\kappa_1) = \frac{1}{2\sqrt{A\pi}}
\int_0^\infty \left(v_1
e^{-\frac{\kappa_1^2 A}{4t}}-v_2 e^{-\frac{\kappa_2^2 A}{4t}}\right)
\left[ \theta_3\left({\rm e}^{-t\Delta}\right)
\theta_3\left({\rm e}^{-t/\Delta}\right) - 1 \right]\frac{{\rm d} t}{\sqrt{t}}.
\end{equation}
Notice that the exponential terms with non-zero $\kappa_1$ and $\kappa_2$
make the integral convergent for small $t$ when
$\theta_3\left({\rm e}^{-t\Delta}\right)
\theta_3\left({\rm e}^{-t/\Delta}\right)\sim 1/t$ as $t\to 0$.

\begin{figure}[tbp]
\centering
\includegraphics[clip,height=8.5cm]{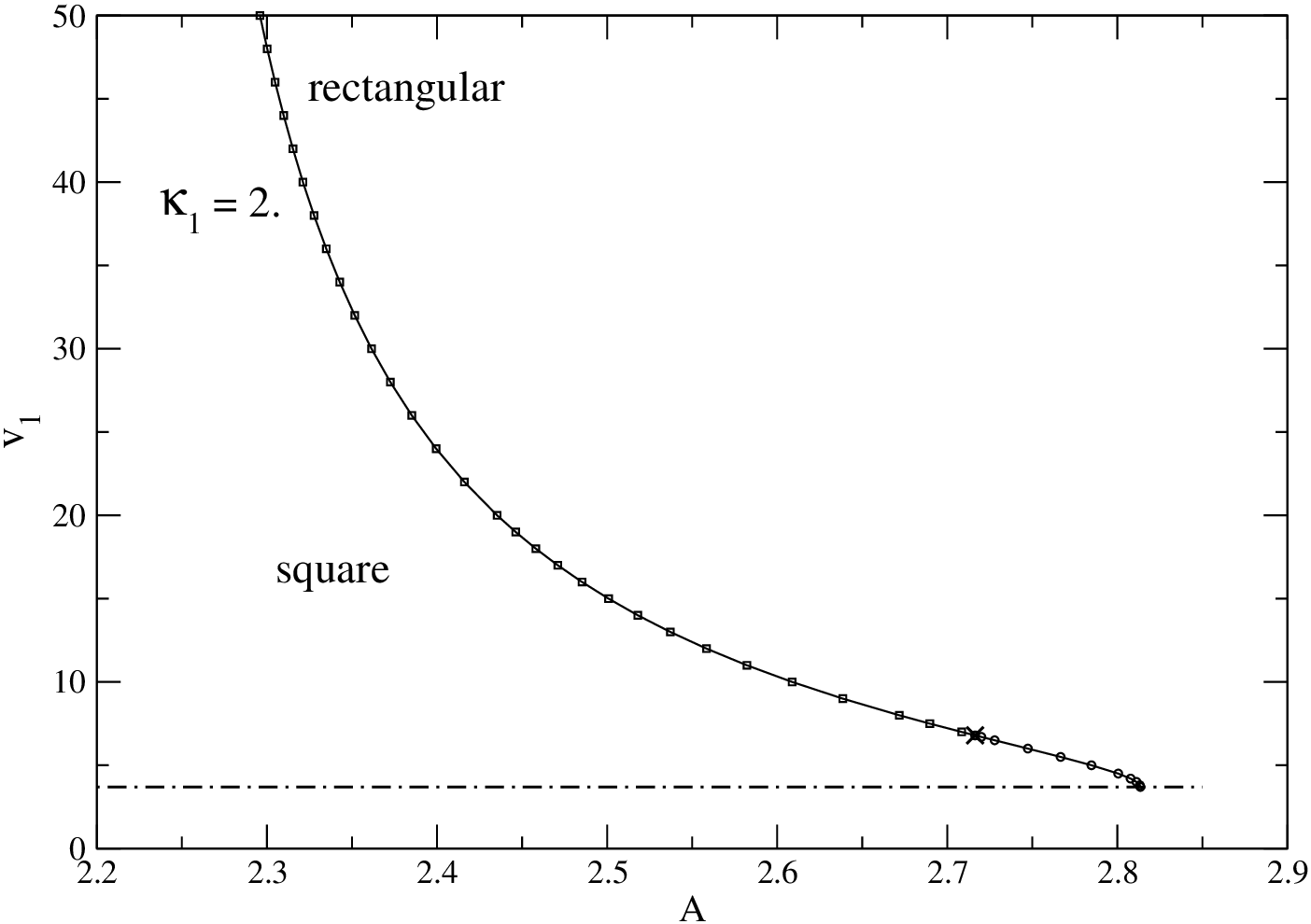}\lapbox[5cm]{-7.2cm}{
\raisebox{3.4cm}{\includegraphics[clip,height=4.9cm]{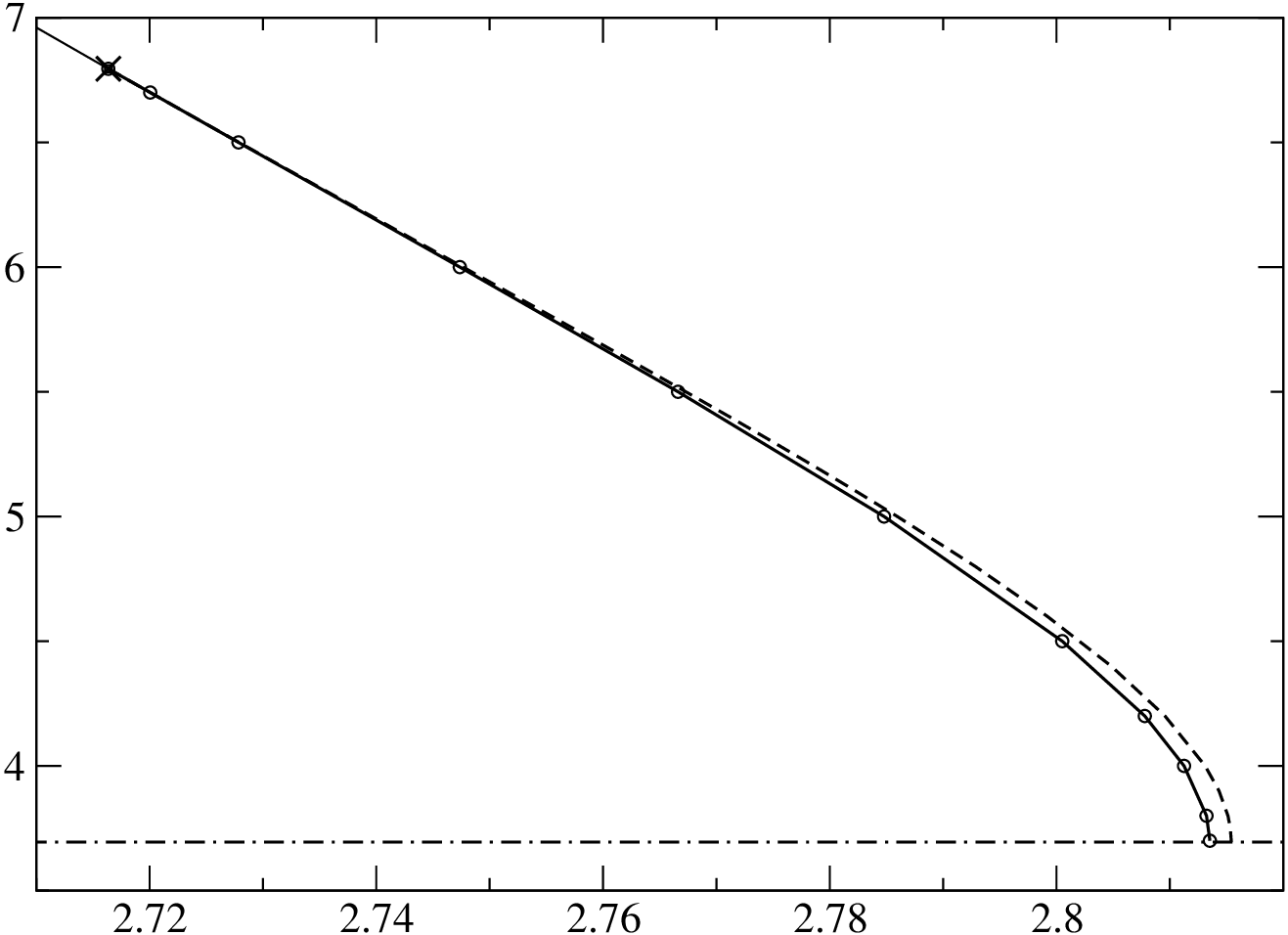}}}
\caption{Phase diagram in the $(A,v_1)$-plane for the double Yukawa
potential with the fixed parameter $\kappa_1=2$.
The critical curve of second-order transitions between the square and
rectangular lattices is marked by open squares interconnected via
a solid line.   
The cross marks the tricritical point.
The discontinuous first-order transitions are indicated by open circles.
The dashed line in the inset, given by $E_2(A)=0$, represents an artificial
extension of the second-order transitions to the region dominated
by first-order transitions.
The dash-dotted straight line $v_1={\rm e}^{\kappa_1}/\kappa_1={\rm e}^2/2$
denotes the border of $v_1$-values.}
\label{fig1}
\end{figure}

\medskip

Let us now show that the system always admits a transition point.

\begin{lemma}[\textbf{Existence of transition points}]\label{lem:transyukawa}
For any fixed $(v_1,\kappa_1)$, $E$ admits a transition point
$A^*(v_1,\kappa_1)$.
\end{lemma}
\begin{proof}
For $t>0$, we have, since $v_1>v_2$ and $\kappa_1>\kappa_2$,
\begin{equation*}
v_1\exp{\left(-\frac{\kappa_1^2}{4t}\right)}
-v_2\exp{\left(-\frac{\kappa_2^2}{4t}\right)}<0 \iff t<
\frac{(\kappa_1^2-\kappa_2^2)}{4\ln\left(\frac{v_1}{v_2} \right)}=:r_0.
\end{equation*}
It follows that $\mu_f<0$ on $(0,r_0)$ and therefore, by Proposition \ref{prop:transition}, $E$ admits a transition point.
\end{proof}
\medskip

\textbf{Numerical investigation.} Let us fix one of the two independent Yukawa parameters, say $\kappa_1=2$. By Lemma \ref{lem:transyukawa}, for any given $v_1$, $E$ admits a transition point $A^*(v_1)$. The phase diagram in the $(A,v_1)$-plane is pictured in Fig. \ref{fig1}. 
Eq. (\ref{transpointenlarged}) was used to calculate the critical curve
of second-order transitions defined by 
$$
\mathcal{C}:=\{(A^*,v_1(A^*)) : A^* \textnormal{ transition point}\}
$$
between the square and rectangular phases, marked by open squares
interconnected via a solid line.
For a fixed $v_1$, the square (rectangular) lattice minimizes the energy
in the whole interval $0\le A\le A^*$ ($A>A^*$). We observe the following:
\begin{enumerate}
\item $A^*\mapsto v_1(A^*)$ is decreasing.
\item There exists a finite ``minimal'' value of $A^*$,
$A^*_{\min}\approx 2.18626$, such that 
$$\displaystyle \lim_{A\to A^*_{\min} \atop A>A^*_{\min}} v_1(A_{\min}^*) = +\infty.
$$
\item The critical curve $\mathcal{C}$ ends up at the tricritical point
with the coordinates $A^t\approx 2.7163619942262467$ and
$v_1^t:=v_1(A^t)\approx 6.7951845011079$, denoted by the cross.
\end{enumerate}
For ${\rm e}^{\kappa_1}/\kappa_1<v_1<v_1^t$, the transition between the square
and rectangular lattices is of the first order, see open circles in
Fig. \ref{fig1}.
For comparison, an artificial prolongation of the critical curve $v_1(A^*)$
into this region is indicated by the dashed line in the inset
of the figure. 

\begin{figure}[tbp]
\centering
\includegraphics[clip,width=0.7\textwidth]{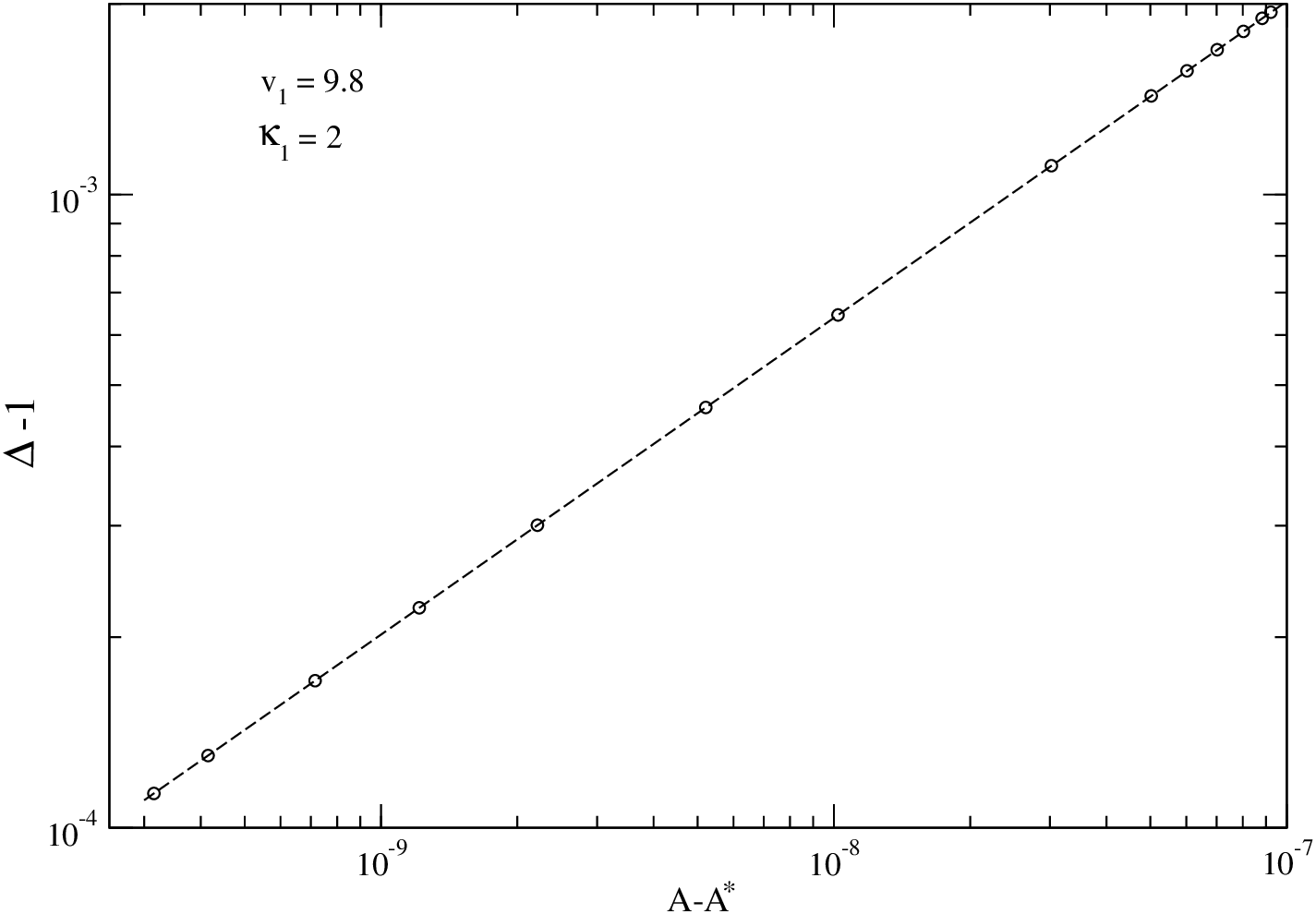}
\caption{The double Yukawa potential with the parameters $\kappa_1=2$ and
$v_1=9.8$, corresponding to the critical value of $A^*\approx 2.61449322978$. 
The log-log plot of $\Delta-1$ versus $A-A^*$ is presented, numerical data
are represented by open circles interconnected via a dashed line.
The fitting via $\Delta-1\propto(A-A^*)^{\beta}$ yields the critical exponent
$\beta^*\approx 0.50003$.}
\label{fig2}
\end{figure}

\medskip

\textbf{Singularity in the $A>A^*$ and $A>A^t$ regimes.} To document the singular behaviour of $\varepsilon\sim \Delta-1$
close to the critical curve, let us choose $\kappa_1=2$ and $v_1=9.8$
and the corresponding critical value of the inverse density
$A^*\approx 2.61449322978$. 
For $A$ slightly larger than $A^*$, the rectangle energy is minimized
with respect to $\Delta$.
The log-log plot of the obtained $\Delta-1$ versus $A-A^*$ 
is presented in Fig. \ref{fig2}.
The data were fitted according to $\Delta-1\propto(A-A^*)^{\beta}$.
The obtained critical exponent $\beta\approx 0.50003$ is very close to
the anticipated mean-field exponent $\beta=1/2$.

\begin{figure}[tbp]
\centering
\includegraphics[clip,width=0.7\textwidth]{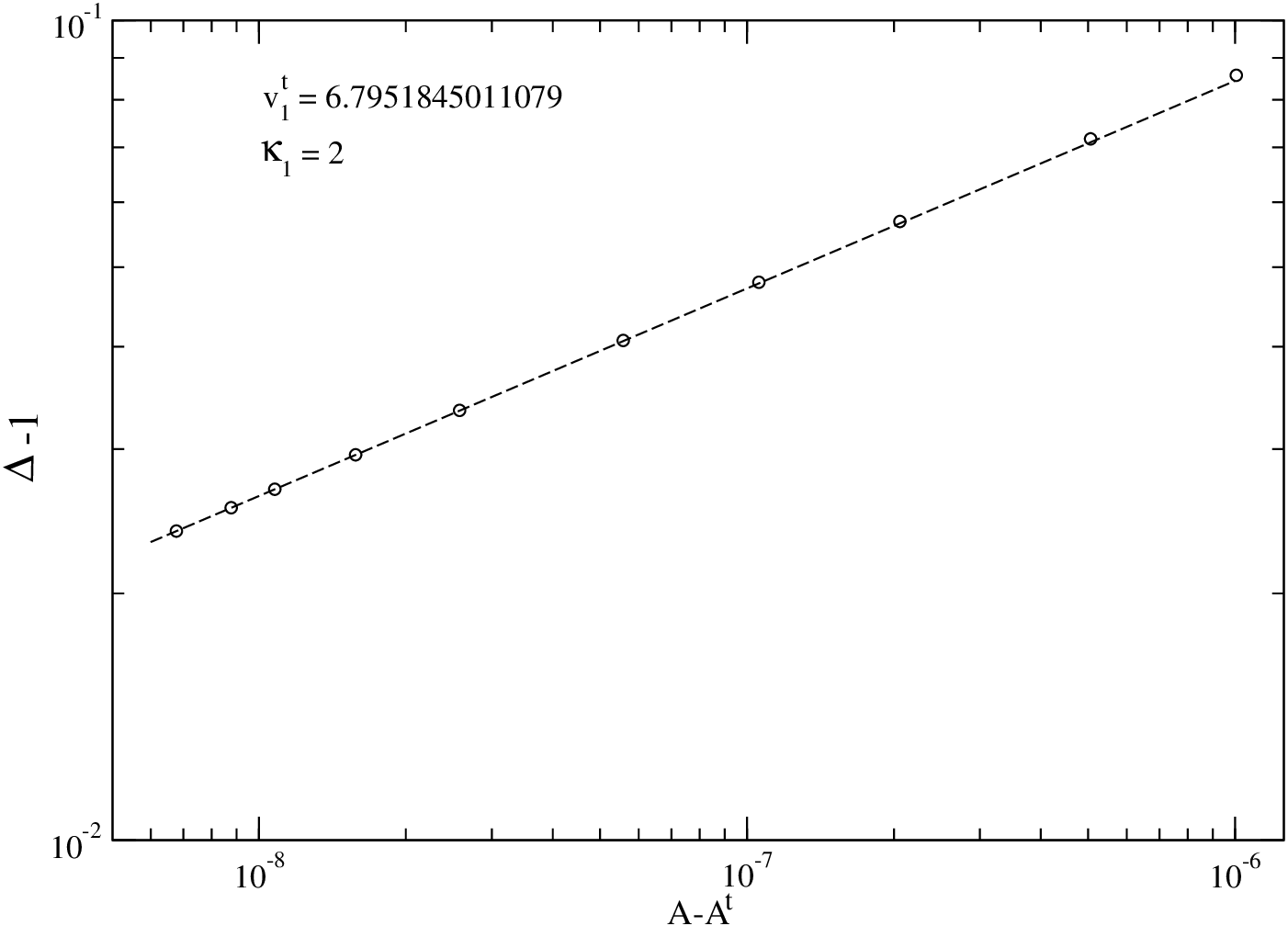}
\caption{The double Yukawa potential with the parameter $\kappa_1=2$,
the corresponding tricritical point has coordinates
$A^t=2.7163619942262467\ldots$ and $v_1^t=6.7951845011079\ldots$.
The log-log plot of $\Delta-1$ versus $A-A^t$ is presented in the
rectangular region $A>A^t$, numerical data are denoted by open circles
interconnected via a dashed line.
The fitting via $\Delta-1\propto(A-A^*)^{\beta}$ yields the tricritical
exponent $\beta^t\approx0.253$.}
\label{fig3}
\end{figure}

As concerns the tricritical point for the double Yukawa parameter $\kappa_1=2$,
see the cross in Fig. \ref{fig1}, its coordinates
$A^t=2.7163619942262467\ldots$ and $v_1^t=6.7951845011079\ldots$
were calculated by using the formula (\ref{tricritical}).
The numerical calculation of the deviation $\Delta-1$ in the region of
the rectangular lattice $A>A_t$ close to the tricritical point was made.
The log-log plot of numerical data in Fig. \ref{fig3} can be fitted as
$\Delta-1\propto(A-A^t)^{\beta^t}$, the obtained exponent $\beta^t\approx0.253$
is reasonably close to the expected mean-field value $\beta^t=1/4$.

\begin{figure}[tbp]
\centering
\includegraphics[clip,width=0.7\textwidth]{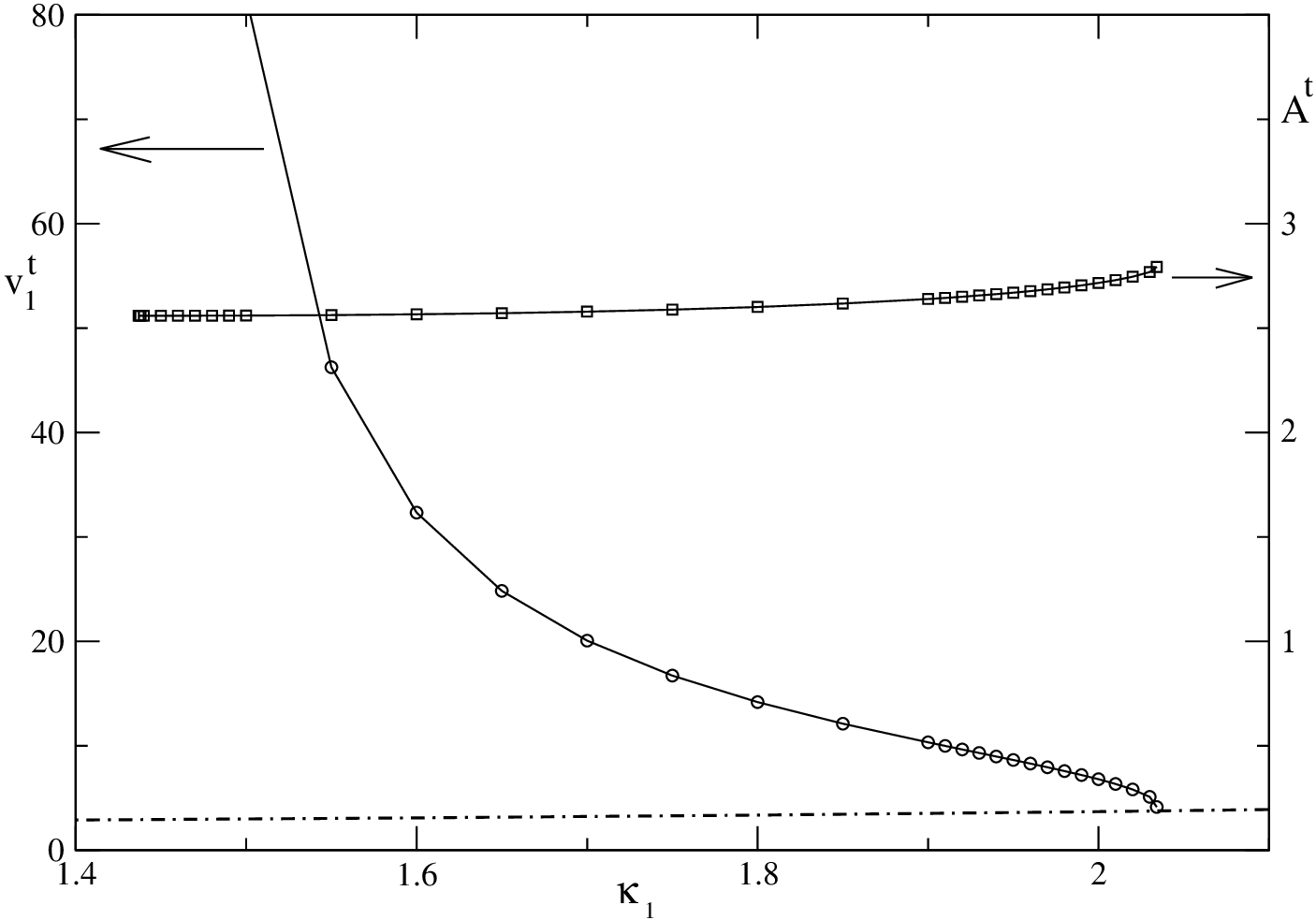}
\caption{The components $(A^t,v_1^t)$ of the tricritical points
as the functions of the Yukawa parameter $\kappa_1$.
The tricritical point exists only if $\kappa_1\in (\kappa_1^L,\kappa_1^U)$
where $\kappa_1^L \approx 1.436$ and $\kappa_1^U\approx 2.03414$ are
the lower and upper limits, respectively.  
The values of $A^t$ are represented by squares and the ones of $v_1^t$
by circles.
The dash-and-dot line corresponds to the curve $v_1={\rm e}^{\kappa_1}/\kappa_1$, 
there are no physical solutions below it.}
\label{fig4}
\end{figure}

\medskip

\textbf{About the transition/tricritical points and their $\kappa_1$-dependence.} The pair of coordinates of the tricritical point $A^t$ and $v_1^t$ 
can be calculated by using Eq. (\ref{tricritical}) for any value of
$\kappa_1$, the obtained results are presented in Fig. \ref{fig4}.
It turns out that a physically meaningful solution for the tricritical point
exists only if $\kappa_1\in (\kappa_1^L,\kappa_1^U)$ where
$\kappa_1^L \approx 1.436$ and $\kappa_1^U\approx 2.03414$ are the lower and
upper limits, respectively.
While $v_1^t$ is a decreasing function of $\kappa_1$, $A^t$ is a slowly
increasing function of $\kappa_1$.
The value of $v_1^t$ diverges when $\kappa_1$ approaches its lower limit
$\kappa_1^L$ from above. 
Consequently, all transitions are of first order (discontinuous) for
$0<\kappa_1\le \kappa_1^L$.
The upper limit $\kappa_1^U$ is given by the intersection of the $v_1^t$-curve
with the dash-and-dot line $v_1={\rm e}^{\kappa_1}/\kappa_1$ which is the border
of the accessible subspace for Yukawa parameters (\ref{subspace}). 
All transitions are of second order (continuous) for $\kappa_1>\kappa_1^U$;
for every $\kappa_1$ from this interval, the critical line goes down up to
the border line $v_1={\rm e}^{\kappa_1}/\kappa_1$, leaving no space for
a tricritical point and the corresponding first-order transitions.

Recall that for $\kappa_1=2$ the ``minimal'' critical value of $A^*$ at which
$v_1(A^*)$ goes to infinity was found numerically to be nonzero, in particular
$A^*_{\min}\approx 2.18626$.
In what follows, we aim at deriving an exact formula for
$A^*_{\min}(\kappa_1)$ for any value of $\kappa_1>0$.
According to (\ref{v2k2}), for a fixed value of $\kappa_1$ and in the limit
of large $v_1$, the Yukawa parameter $\kappa_2$ behaves as
\begin{equation} \label{asymp1}
\kappa_2 = \kappa_1 - \frac{1}{v_1} (1+\kappa_1) {\rm e}^{\kappa_1}
+ O\left( \frac{1}{v_1^2} \right)   
\end{equation}
and the Yukawa parameter $v_2$ is given by 
\begin{equation} \label{asymp2}
v_2 = v_1 - \kappa_1 {\rm e}^{\kappa_1} + O\left( \frac{1}{v_1} \right) .    
\end{equation}
Since
\begin{equation}
{\rm d}\mu_f\left(\frac{t}{A} \right) = \rho_f\left(\frac{t}{A} \right) \frac{{\rm d}t}{A}
= \left( v_1 e^{-\frac{\kappa_1^2 A}{4t}}
-v_2 e^{-\frac{\kappa_2^2 A}{4t}}\right) \frac{{\rm d}t}{\sqrt{\pi t A}} ,
\end{equation}
using the asymptotic expansions (\ref{asymp1}) and (\ref{asymp2}),
the expression into brackets can be expressed as  
\begin{equation}
v_1 e^{-\frac{\kappa_1^2 A}{4t}} -v_2 e^{-\frac{\kappa_2^2 A}{4t}}
= \kappa_1 {\rm e}^{\kappa_1} e^{-\frac{\kappa_1^2 A}{4t}}
\left[ 1 - \frac{(1+\kappa_1)}{2 t} A \right] +
O\left( \frac{1}{v_1} \right) .
\end{equation}  
With regard to Eq. (\ref{E2A}), the critical condition $E_2(A^*)=0$
takes in the limit $v_1\to\infty$ the form
\begin{equation} 
0 = \int_0^\infty \sqrt{t} e^{-\frac{\kappa_1^2 A^*_{\min}}{4t}}
\left[ 1 - \frac{(1+\kappa_1)}{2 t} A^*_{\min} \right]
\left[ \theta_3 \theta_3^{(1)} - t \left( \theta_3^{(1)}
+ t\theta_3 \theta_3^{(2)} \right)^2 \right] {\rm d}t .
\end{equation}
This equation determines for any $\kappa_1>0$ the exact value of
the critical inverse density $A^*_{\min}(\kappa_1)$ at which $v_1\to\infty$.
In particular, at $\kappa_1=2$ the obtained $A^*_{\min}\approx 2.186262818188$
agrees with the previous numerical estimate $A^*_{\min}\approx 2.18626$.
The monotonous decay of $A^*_{\min}$ with increasing the Yukawa parameter
$\kappa_1>0$ is presented in Fig. \ref{fig5} by the solid line
connecting data (open circles).
The function $A^*_{\min}(\kappa_1)$ tends to unity in the limit
$\kappa_1\to\infty$.
In the opposite limit $\kappa_1\to 0^+$,
\begin{equation} 
\lim_{\kappa_1\to 0^+} A^*_{\min}(\kappa_1) =
\frac{ 2 \int_0^\infty  \sqrt{t}
\left[ \theta_3 \theta_3^{(1)} + t\theta_3 \theta_3^{(2)}
- t \left( \theta_3^{(1)} \right)^2 \right]{\rm d}t}{\int_0^\infty 
\frac{1}{\sqrt{t}}
\left[ \theta_3 \theta_3^{(1)} + t\theta_3 \theta_3^{(2)}
- t \left( \theta_3^{(1)} \right)^2 \right]{\rm d}t}\approx 5.71344 .
\end{equation}  

\begin{figure}[tbp]
\centering
\includegraphics[clip,width=0.7\textwidth]{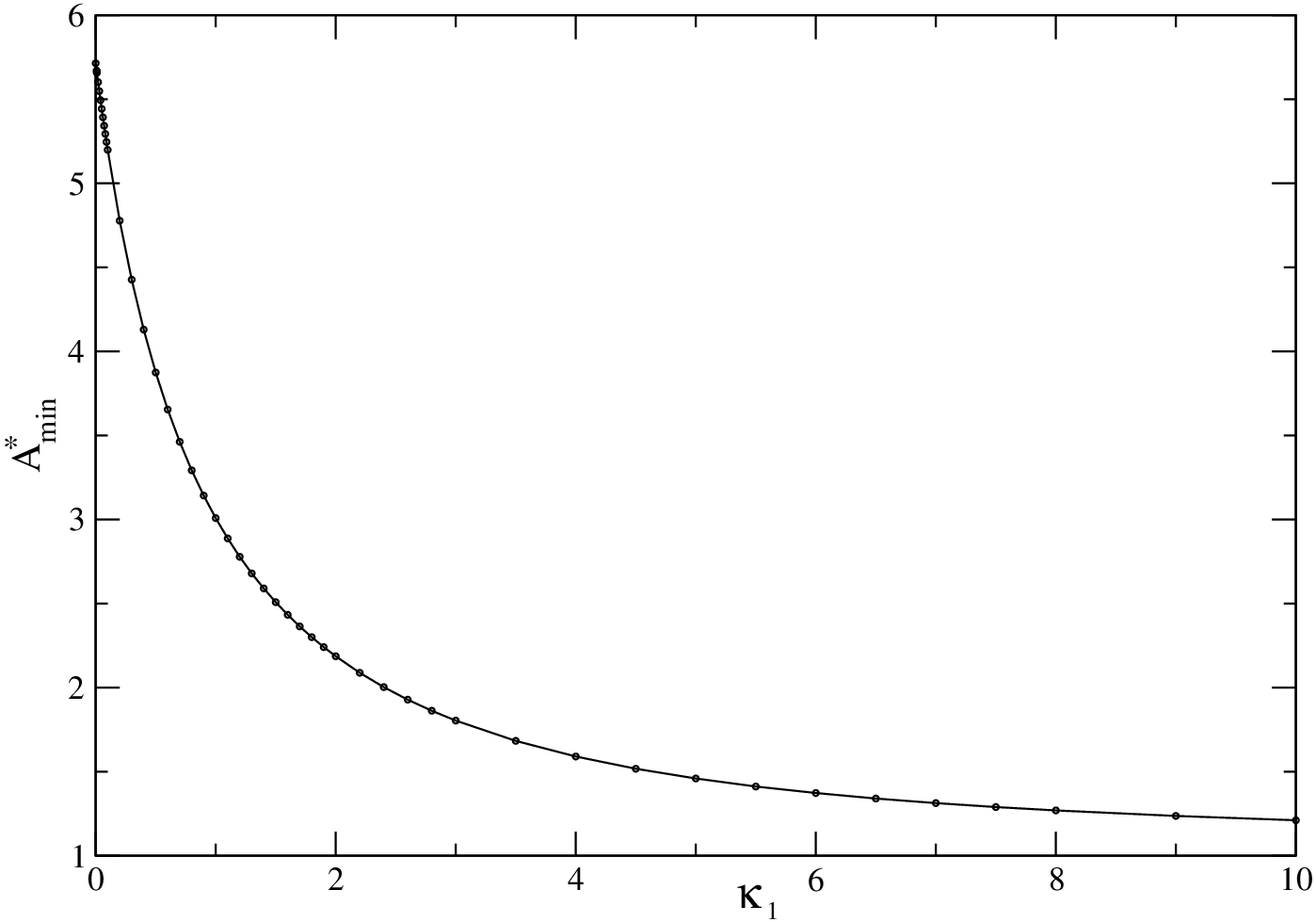}
\caption{The dependence of the minimal value of the critical inverse
density $A^*_{\min}$, at which a second-order transition takes place
in the limit of $v_1\to\infty$, on the Yukawa parameter $\kappa_1>0$.
The limiting values of $A^*_{\min}(\kappa_1)$ are:
$A^*_{\min}(\kappa_1\to\infty) = 1$ and
$A^*_{\min}(\kappa_1\to 0^+)=5.71344\ldots$.}
\label{fig5}
\end{figure}

\renewcommand{\theequation}{5.\arabic{equation}} 
\setcounter{equation}{0}

\section{Yukawa-Coulomb potential} \label{Sec5}
\textbf{Potential, parameters constraints and energy.}
In the special Yukawa-Coulomb case $\kappa_2=0$,
the potential
\begin{equation}
f(r) = v_1 \frac{\exp{(-\kappa_1 r)}}{r} - v_2 \frac{1}{r}
\end{equation}
has a minimum at $r_{\min}=1$ and $f(1)=-1$ under conditions
\begin{equation} \label{v1v2}
v_1 = \frac{e^{\kappa_1}}{\kappa_1} , \qquad
v_2 = \frac{1+\kappa_1}{\kappa_1} .
\end{equation}
One is left with only one independent parameter $\kappa_1>0$
and there are no constraints on this parameter.
The associated measure is therefore
\begin{equation}
{\rm d}\mu_f(t) = 
\left[v_1\exp{\left(-\frac{\kappa_1^2}{4t}\right)} - v_2\right]\frac{{\rm d}t} {\sqrt{\pi t}} .
\end{equation}

It is clear that $f\not\in \mathcal{F}$ since it is not integrable at infinity. Therefore, a regularization of the divergent lattice sum by a neutralizing background
is necessary for the Coulomb term (see \cite{Travenec22} and Remark \ref{rmk-renorm}).
The energy then reads as
\begin{eqnarray}
E(A,\Delta;\kappa_1) & = & \frac{1}{2\sqrt{A\pi}}
\int_0^\infty\Bigg\{ v_1
{\rm e}^{-\frac{\kappa_1^2 A}{4t}} \left[ \theta_3\left({\rm e}^{-t\Delta}\right)
\theta_3\left({\rm e}^{-t/\Delta}\right) - 1 \right] \nonumber \\
& & - v_2\left[ \theta_3\left({\rm e}^{-t\Delta}\right)
\theta_3\left({\rm e}^{-t/\Delta}\right) - 1 -\frac{\pi}{t}\right] \Bigg\}\frac{{\rm d} t}{\sqrt{t}}.
\end{eqnarray}
The neutralizing background manifests itself as the addition of the singular
term $-\pi/t$ in the last square bracket, which ensures the convergence of
the integral. It is again straightforward to show that transitions point exists by direct application of Proposition \ref{prop:transition} in the same way that we did in Lemma \ref{lem:transyukawa} for the $\kappa_2>0$ case.

\begin{figure}[tbp]
\centering
\includegraphics[clip,width=0.7\textwidth]{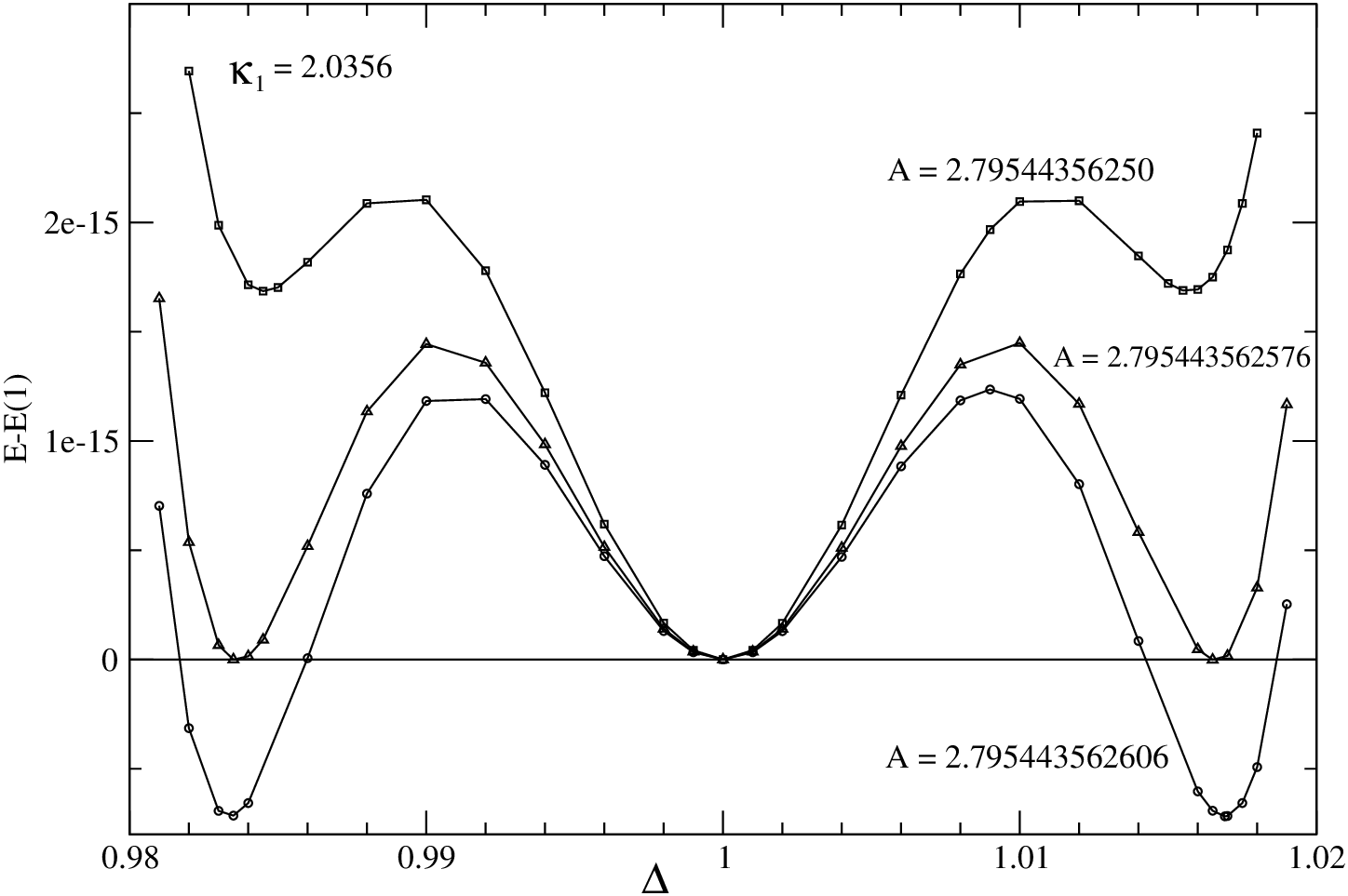}
\caption{The plot of the energy difference $E(\Delta)-E(1)$ versus $\Delta$
for the Yukawa parameter $\kappa_1=2.0365$, just below its tricritical
value $\kappa_1^t=2.036517758847\ldots$ where a first-order transition
between the square and rectangular lattices exists.
Three values of the inverse particle density $A$ are considered:
$A=2.79544356250$ when the square lattice with $\Delta=1$ prevails,
the first-order transition value $A_{\rm trans}=2.795443562576$ where
the energies of the square and rectangular lattices, separated by
energy barriers, coincide and $A=2.795443562606$ where the rectangular
case is dominant via a jump
in $\Delta$.}
\label{fig6}
\end{figure}

\medskip

\textbf{Numerical investigation.} Solving the couple of equations (\ref{tricritical}) with $\kappa_2=0$,
$v_1$ and $v_2$ given by (\ref{v1v2}),
one gets the tricritical point for the Yukawa-Coulomb interaction:
$A^t\approx 2.795433950879$ and $\kappa_1^t\approx 2.036517758847$.
Let us make a small step into the region $\kappa_1<\kappa_1^t$, where
a step-wise first-order transition exists, say $\kappa_1=2.0365$. We observe the following:
\begin{enumerate}
\item As is seen in Fig. \ref{fig6}, there is just one energy minimum at
$\Delta=1$ for $A$ slightly below the first-order transition value
$A_{\rm trans}=2.795443562576$ and the square lattice prevails.
\item For $A$ slightly larger than $A_{\rm trans}$, one gets two equivalent minima with
$\Delta>1$, related via the symmetry $\Delta\to 1/\Delta$, meaning that
the rectangular case is dominant via a jump in $\Delta$.
\item  Exactly at $A_{\rm trans}$ we have three equivalent minima,
i.e., the energies of the square and rectangular lattices, separated by
energy barriers, coincide. 
\end{enumerate}

\begin{figure}[tbp]
\centering
\includegraphics[clip,width=0.7\textwidth]{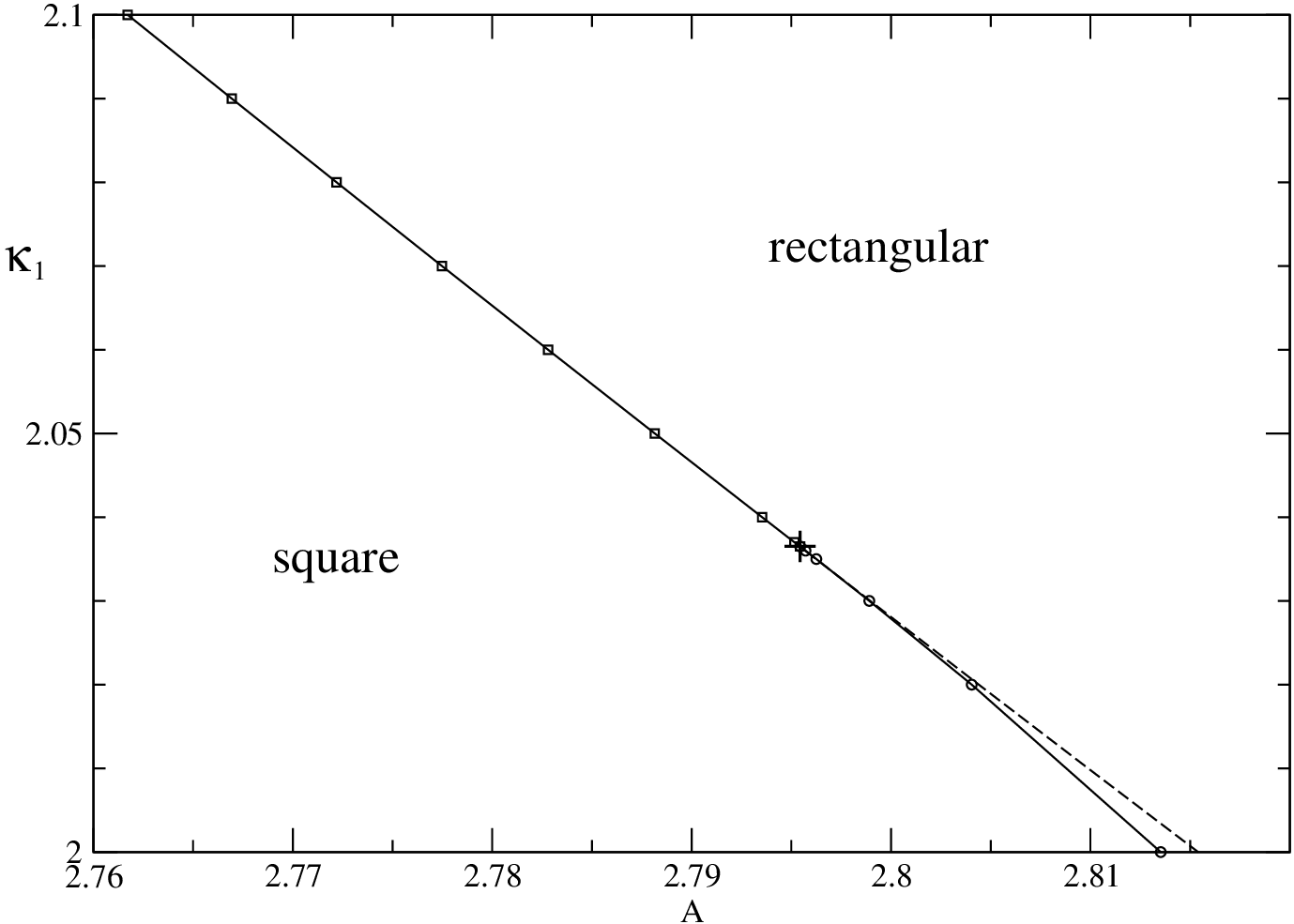}
\caption{Phase diagram of the Yukawa-Coulomb model close to
the tricritical point denoted by the cross.
The critical line of second-order transitions is denoted by open squares,
first-order transitions are indicated by open circles.
The dashed line corresponds to an artificial prolongation of the
critical line of second-order transitions to the region where
first-order transitions are dominant.}
\label{fig7}
\end{figure}

The phase diagram of the Yukawa-Coulomb model is pictured in Fig. \ref{fig7}.
The second-order transitions are marked by open squares, the first-order
by open circles and the tricritical point by the cross.
The dashed line is an artificial prolongation of the critical line
beyond the tricritical point.

\begin{figure}[tbp]
\centering
\includegraphics[clip,width=0.7\textwidth]{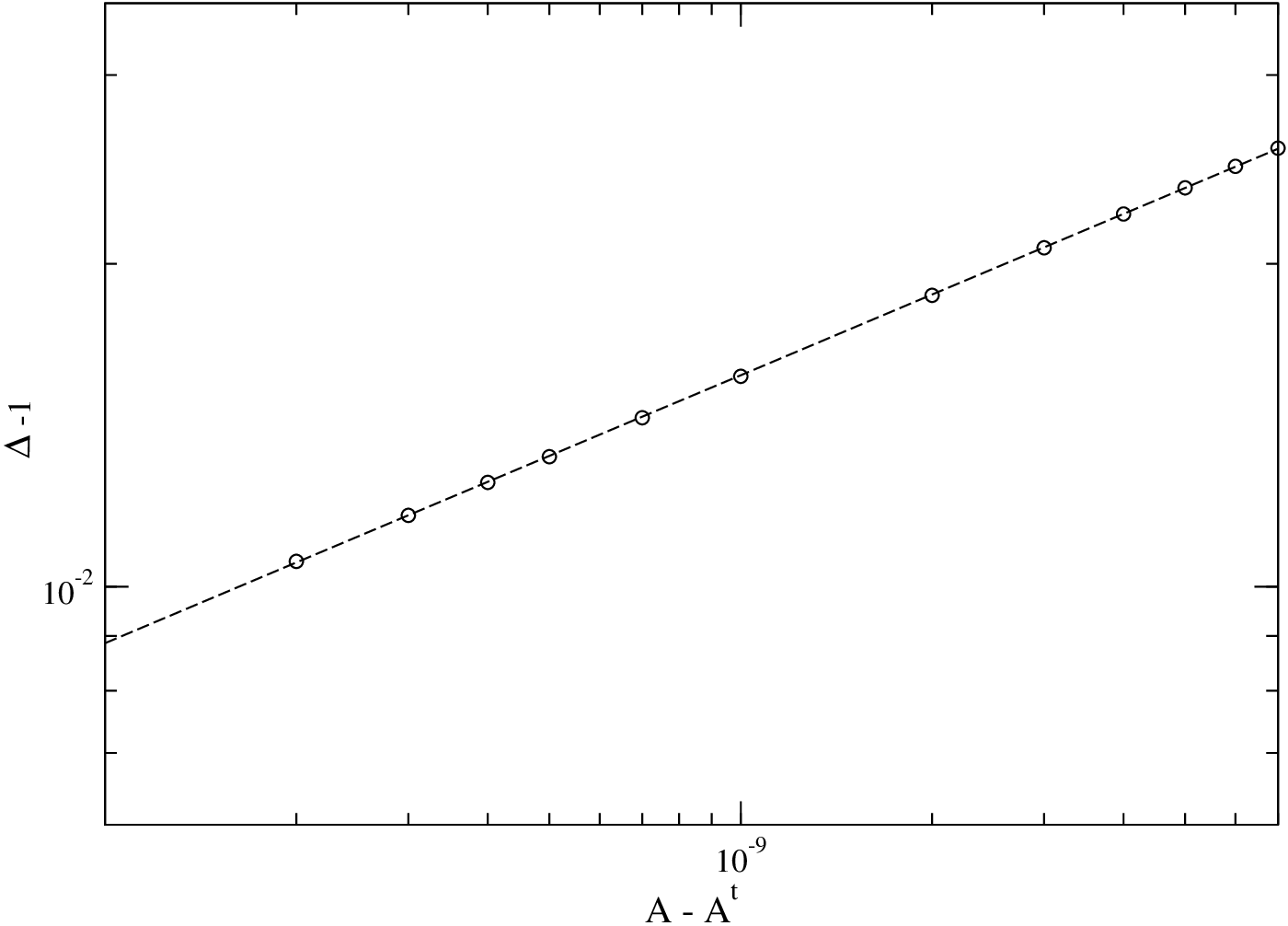}
\caption{The plot of the deviation $\Delta-1$ versus $A-A^t$
for the Yukawa-Coulomb model close to the tricritical point.}
\label{fig8}
\end{figure}

In Fig. \ref{fig8}, $\Delta-1$ is plotted as the function of $(A-A^t)$
in order to check the tricritical exponent in the Yukawa-Coulomb case.
The numerical data are represented by open circles.
Fitting the log-log plot via $\Delta-1 \propto (A-A^t)^{\beta^t}$, represented
by the dashed line, yields the tricritical exponent $\beta^t=0.2498$ which
is very close to the anticipated mean-field one $\beta^t=1/4$.

\section*{Acknowledgment}
The support received from the project EXSES APVV-20-0150
and VEGA Grant No. 2/0092/21 and is acknowledged.

\end{document}